\documentclass[11pt]{amsart}
\usepackage{amsmath, amssymb}
\usepackage{upgreek,bm}
\usepackage{xcolor, url}
\usepackage{todonotes} 
\usepackage[all]{xy}
\usepackage{graphicx}
\usepackage{bm}
\usepackage{enumerate}
\theoremstyle{plain}
\newtheorem{theorem}{Theorem}[section]
\newtheorem{lemma}[theorem]{Lemma}

\newtheorem{proposition}[theorem]{Proposition}
\newcommand{\td}{\mathrm{d}}
\theoremstyle{definition}

\newtheorem{definition}[theorem]{Definition}

\numberwithin{equation}{section}
\newcommand{\tr}{\operatorname{Tr}}

\usepackage{subfig, tikz}
\usetikzlibrary{decorations.pathmorphing}
\usepackage{graphics}

\usepackage{hyperref} %added
\hypersetup{colorlinks=true, linkcolor=blue, citecolor = blue, urlcolor=blue, filecolor=magenta}

\begin{document}
\title[The Penrose inequality in spherical symmetry with charge and in Gauss-Bonnet gravity]{The Penrose inequality in spherical symmetry with charge and in Gauss-Bonnet gravity}
\author{Hari K. Kunduri}
\address{Department of Mathematics and Statistics and Department of Physics and Astronomy\\
		McMaster University\\
		Hamilton, ON Canada}
	\email{kundurih@mcmaster.ca}
\author{Juan Margalef-Bentabol}
\address{Departamento de Matemáticas\\
Universidad Carlos III de Madrid\\
Leganés, Spain\\
Grupo de Teorías de Campos y Física Estadística, Instituto Gregorio Millán\\
Unidad Asociada al Instituto de Estructura de la Materia CSIC\\
Madrid, Spain\\
\newline
Department of Mathematics and Statistics\\
		Memorial University\\
		St John's, NL Canada
}
	\email{juan.margalef@uc3m.es}
\author{Sarah Muth}
\address{Department of Mathematics and Statistics\\
		Memorial University\\
		St John's, NL Canada}
	\email{smmmuth@mun.ca}

%\date{\today}

\thanks{H. K. Kunduri acknowledges the support of the  NSERC Grant  RGPIN-2018-04887. J.  Margalef-Bentabol acknowledges the support of the PID2020-116567GB-C22 grant funded by MCIN/AEI/10.13039/501100011033}

\begin{abstract}
We establish the spacetime Penrose inequality in spherical symmetry in spacetime dimensions $n+1\geq3$ with charge and cosmological constant from the initial data perspective. We also show that this result extends to the Gauss-Bonnet theory of gravity.
\end{abstract}

\maketitle

\section{Introduction}
\noindent The Penrose inequality \cite{Penrose} is a conjectured inequality which asserts that the mass of an asymptotically flat black hole spacetime is bounded below by a function of the area of a spatial cross-section of its event horizon (see the comprehensive review \cite{Mars}). The bound is saturated if and only if the spacetime is isometric to the Schwarzschild solution. The conjecture is based upon heuristic arguments arising from our physical model of gravitational collapse (the weak cosmic censorship conjecture and final state conjecture \cite{Klainerman}). A counterexample would cast doubt on at least one of these two central conjectures. Generalizations of the inequality that incorporate charge and angular momentum are expected to hold \cite{Dain:2013qia}.

The conjecture can be stated precisely in terms of initial data set $(M, g, k)$ where $(M,g)$ is a Riemannian manifold of dimension $n$ and $k$ is a symmetric rank-2 tensor. The triplet satisfies the appropriate constraint equations (see Equations \eqref{eq: constraint EM} and \eqref{eq: constraints GB} below) for some energy-momentum density $\mu,J$, which encode the matter fields present in the spacetime. They are assumed to satisfy the dominant energy condition $\mu\geq|J|_g$. 

\subsection{Asymptotically flat case}\mbox{}\vspace{.5ex}

\noindent The initial data set is asymptotically flat if there exists a compact set $K \subset M$ and a diffeomorphism $\Phi: M \setminus K \to \mathbb{R}^n \setminus B$ where $B$ is a closed coordinate ball such that in the coordinates $\{x^i\}$ associated to $\Phi$, 
\begin{equation}
    \Phi^* g = \delta + O_2(r^{-\kappa}), \qquad \Phi^* k = O_1(r^{-\kappa-1}), \qquad 
 \Phi^*\mu,  \Phi^*J=O(r^{-2\kappa}), 
\end{equation} for some $\kappa > (n-2)/2$. We have the well-defined geometric invariants known as ADM energy and linear momenta which are defined by
\begin{subequations}
\begin{align}
&E_{\mathrm{ADM}}  = \frac{1}{2(n-1) \omega_{n-1}} \lim_{r \to \infty} \int_{\mathbb{S}_r^{n-1}} (\partial_i g_{ij} - \partial_j g_{ii})\nu^j \; \mathrm{vol}(\mathbb{S}_r^{n-1})\label{eq: EADM}\\
&P_i =\frac{1}{(n-1)\omega_{n-1}}\lim_{r\rightarrow\infty}
\int_{\mathbb{S}_r^{n-1}}(k_{ij}-(\text{Tr}_g k)g_{ij})\nu^{j}\; \mathrm{vol}(\mathbb{S}_r^{n-1}),
\end{align}
\end{subequations} 
where $\omega_{n}$ denotes the volume of the unit sphere $\mathbb{S}^n$, $\nu^j$ denotes the unit-normal vector field to the $n$-sphere $\mathbb{S}_r^{n-1}$ of radius $r$ and $\mathrm{vol}(\mathbb{S}_r^{n-1}) = r^{n-1} \mathrm{vol}(\mathbb{S}^{n-1})$. The ADM mass $m$ is then defined as $m = \sqrt{E^2 - |\mathbf{P}|}^2$ where $\mathbf{P}$ is the momentum $n$-vector associated with the initial data set. By the positive mass theorem, $m\geq 0$ with equality if and only if the initial data can be embedded into Minkowski spacetime \cite{SY1,SY2,Witten}. A quasi-local-in-time proxy for the horizon of a black hole is a marginally outer trapped surface (MOTS). This is a hypersurface $\Sigma \subset M$ (co-dimension 2 in spacetime) which has the property that the outward null expansion $\theta_+$ vanishes, where the outward and inward null expansions $\theta_\pm$ are given by 
\begin{equation}
    \theta_\pm = H_\Sigma \pm \tr_\Sigma k,
\end{equation} where $H_\Sigma$ denotes the mean curvature of $\Sigma$ within $M$. The notion of a MOTS captures the intuitive idea that outward-directed light rays cannot escape $\Sigma$. The \emph{outermost MOTS} (often referred to as an apparent horizon) is defined to be a MOTS that is not enclosed within any other MOTS. The Penrose conjecture can then be stated as 
\begin{equation}\label{PIvac}
    m \geq \frac{1}{2} \left(\frac{A}{\omega_{n-1}}\right)^{\frac{n-2}{n-1}}.
\end{equation} Here $A$ is understood to be the smallest area required to enclose the apparent horizon. In the present work, as we will explain in detail later, $A$ may be identified with the area of the apparent horizon itself. Here equality is to be achieved if and only if the initial data is isometric to a spatial hypersurface in the Schwarzschild spacetime in dimension $n+1$. 

The conjecture has been rigorously established in $n = 3$ for time-symmetric ($k=0$) initial data sets (this is called the `Riemannian-Penrose inequality') in the seminal work of Huisken and Ilmanen \cite{Huisken:2001} and Bray \cite{Bray} using inverse mean curvature flow and conformal flow respectively. These results were subsequently extended to $n \leq 7$ by Bray-Lee \cite{Bray:2007opu}.  For $k \neq 0$, the problem is much more difficult. Progress has been made in the spherically symmetric setting (for a clear exposition, see \cite[Theorem 7.46]{Leetext}) and more recently for cohomogeneity-one initial data \cite[Theorem 1.1]{KhuriKunduri} using the generalized Jang equation approach developed by Bray-Khuri~\cite{BrayKhuri}.  In all these cases, there is an associated rigidity statement, namely that the equality is satisfied if and only if the initial data set is isometric to Schwarzschild initial data.

\subsection{Asymptotically AdS case}\mbox{}\vspace{.5ex}

\noindent In the presence of a negative cosmological constant, it is natural to consider spacetimes that are asymptotically Anti-de Sitter (AdS$_{n+1}$). The initial data $(M,g,k)$ associated with spacelike hypersurfaces of asymptotically AdS spacetimes are asymptotically hyperbolic. More precisely, let $(\mathbb{H}^n, b)$ denote  $\mathbb{R}^n$ with the hyperbolic metric
\begin{equation}
    b = \frac{\td r^2}{V_{\mathrm{hyp}}} + r^2 g_{\scriptscriptstyle \mathbb{S}}.
\end{equation}
where $V_{\mathrm{hyp}}(r):=1+r^2$. The metric, scaled so that $\text{Ric}(b) = -(n-1) b$, models constant time spacelike hypersurfaces in asymptoticallyAdS$_{n+1}$ with spacetime metric $\mathbf{g} = -V_{\mathrm{hyp}}\; \td t^2 + b$ and negative cosmological constant $\Lambda = - n(n-1)/2$ (equivalent to fixing the AdS length scale to $\ell =1$).

An initial data set $(M,g,k)$ is asymptotically hyperbolic if there is a compact set $K \subset M$ and a diffeomorphism $\Phi: M^n \setminus K \to \mathbb{H}^n \setminus B$ where $B$ is a closed coordinate ball such that in the asymptotic chart 
\begin{equation}\label{eq: hatg}
    \hat g = \Phi^* g - b = O_2 (r^{-\kappa}), \qquad \Phi^* k = O_1(r^{-\kappa}),  \qquad 
 \Phi^*\mu,  \Phi^*J=O(r^{-2\kappa}), 
\end{equation} for $\kappa > n/2$. A well-defined notion of energy valid in this setting is the Chrusciel-Nagy energy \cite{Chrusciel:2001qr}, which relies on the so-called static potential (or static lapse function) $W:=\sqrt{V_{\mathrm{hyp}}}$:
\begin{equation}\label{eq: Ehyp}
E_{\mathrm{hyp}} =\frac{1}{2(n-1) \omega_{n-1}} \lim_{r \to \infty} \int_{\mathbb{S}_r^{n-1}} \left[W\,\text{div}_b \hat{g} - W \td (\mathrm{Tr}_b \hat{g}) + (\mathrm{Tr}_b \hat{g}) \td W - \hat{g}(\mathcal{D}W, ~)\right]\!(\mathcal{V}) \mathrm{vol}(\mathbb{S}_r^{n-1}),
\end{equation} 
where $\mathcal{V}^i=W\partial_r^i$ is the unitary outward $b$-normal vector to the coordinate spheres $\{r=r_0\}$ and $\mathcal{D}$ the $b$-Levi-Civita connection.

An extension of \eqref{PIvac} valid in the AdS setting is expected to hold, which is of particular interest from the holographic perspective \cite{EH}. The inequality has been established in spherical symmetry (\cite[Theorem 6]{EF}, \cite[Theorem 1]{Folkestad}, \cite{HS}) with additional various hypotheses. We will recover these results from the initial data perspective along the lines of \cite{Leetext}. More generally, the Penrose inequality in the asymptotically hyperbolic setting has been established for small time symmetric perturbations of Schwarzschild-AdS data~\cite{Ambrozio:2014rfi} and for graphs~\cite{Dahl:2012yh}. The spacetime Penrose inequality for the wider class of cohomogeneity-one initial data has also been established using the Jang equation approach \cite[Theorem 1.2]{KhuriKunduri}.

In this note, we will focus on establishing two extensions of the spacetime Penrose inequality in spherical symmetry: the Einstein-Maxwell theory and the Gauss-Bonnet theory.

\subsection{Einstein Maxwell initial data}\mbox{}\vspace{.5ex}

\noindent Consider \emph{charged initial data sets} $(M,g,k,E,B)$ in $n$ dimensions, which satisfy the constraints arising from Einstein-Maxwell theory with additional (uncharged) matter sources. Here $(E,B)$ denote the electric field one-form and magnetic field $n-2$ form (see Section \ref{section EM} for further details). Under appropriate fall-off conditions on the electric field, we define the electric charge (both in the asymptotically flat and asymptotically AdS) as
\begin{equation}\label{eq: charge}
q = \left(\frac{(n-2)(n-1)}{2}\right)^{-1/2} \frac{1}{\omega_{n-1}} \lim_{r \to \infty} \int_{\mathbb{S}_r^{n-1}} E(\nu)\; \mathrm{vol}(\mathbb{S}_r^{n-1}).
\end{equation}
Consider the asymptotically flat setting in dimension $n=3$. The charged Riemannian-Penrose inequality ($k=0$ but without any symmetry assumptions, and for multiple black holes) has been rigorously established ~\cite{Khuri:2014wqa}. The spacetime with $k\neq 0$ case in spherical symmetry was established in~\cite{DK} using a Jang equation approach. In higher dimensions $n>3$, the charged Riemannian-Penrose inequality was proved for Einstein–Maxwell initial data sets, which could be embedded as a hypersurface in Euclidean space $\mathbb{R}^{n+1}$~\cite{LopesdeLima:2014usy}.
We prove 
\begin{theorem}\label{Thm:PI} Consider a spherically symmetric charged asymptotically flat/hyperbolic initial data set $(M,g,k,E,B)$ of the Einstein-Maxwell equations with uncharged matter sources (with a negative cosmological constant in the latter case) satisfying the constraint equations \eqref{eq: constraint EM} and the dominant energy condition. Then 
\begin{equation}\label{eq: PenroseInequality}
      \mathcal{E} \geq \frac{1}{2} \left[ \left(\frac{A}{\omega_{n-1}}\right)^{\frac{n-2}{n-1}} + q^2 \left(\frac{\omega_{n-1}}{A}\right)^{\frac{n-2}{n-1}} + \frac{1}{\ell^2}\left(\frac{A}{\omega_{n-1}}\right)^{\frac{n}{n-1}} \right],
\end{equation} where $\mathcal{E}$ is the ADM mass (asymptotically flat case, $\ell \to \infty$) or the Chrusciel-Nagy energy $E_{\mathrm{hyp}}$ (asymptotically hyperbolic case, $\ell = 1$) respectively,  $q$ the charge \eqref{eq: charge}, and $A$ the area of the apparent horizon. Moreover, equality holds if and only if the initial data can be isometrically embedded as a spatial hypersurface in the Reissner-Nordstr\"om(-AdS$_{n+1}$) spacetime of mass $\mathcal{E}$ and charge $q$. 
\end{theorem} Here, by spherically symmetric, we mean that $M$ is diffeomorphic to $[0,\infty) \times \mathbb{S}^{n-1}$ with boundary $\partial M$ that is an outermost MOTS (see subsection \ref{SSdata} below). The asymptotically flat case of this result was proved in \cite{Hidayat:2016iez} under the assumption that the initial data set is maximal $(\tr k =0)$.

\subsection{Gauss-Bonnet gravity}\mbox{}\vspace{.5ex}

\noindent Our second result concerns the spacetime Penrose inequality in spherical symmetry within the setting of Einstein Gauss-Bonnet gravity (the reader is referred to the recent review \cite{Fernandes:2022zrq} and references therein). There is a vast variety of proposed modifications and generalizations of general relativity, typically including additional fundamental fields and coupling parameters. A particular mathematically natural generalization is Einstein–Gauss–Bonnet theory, first proposed by Lanczos \cite{Lanczos} and later generalized by Lovelock \cite{Lovelock1,Lovelock2}. The latter's family of theories does not introduce any new fields and is unique in preserving the property that the field equations only contain up to second derivatives of the metric tensor. In particular, Lovelock showed that in spacetime dimension 4, the Einstein-Hilbert functional is the unique choice which gives rise to a left-hand side of the field equations (i.e. not involving matter fields) that is symmetric, divergenceless, and involves no more than second derivatives of the metric. In dimensions 5 and 6, one may add an additional term $L_{GB}$ to the Einstein-Hilbert action consisting of a certain combination of scalar invariants while preserving these properties:
\begin{equation}\label{GBint}
    L_{\mathrm{GB}}:= \mathbf{R}^2 - 4 \mathbf{R}_{ab} \mathbf{R}^{ab} + \mathbf{R}_{abcd}\mathbf{R}^{abcd}, 
\end{equation} where $\mathbf{R}, \mathbf{R}_{ab}, \mathbf{R}_{abcd}$ denote respectively the scalar curvature, Ricci tensor, and Riemann tensor of the spacetime metric $\mathbf{g}$. \eqref{GBint} is referred to as the `Gauss-Bonnet term'. It is a topological invariant in four spacetime dimensions and does not contribute to the equations of motion (its variation is a pure divergence). In dimensions greater than 6, Lovelock's theorem allows for additional possibilities beyond the Gauss-Bonnet term, but we will restrict attention to \eqref{GBint} in the present work. 

An obvious question is whether the chain of arguments 
which led to the Penrose conjecture continues to hold within Gauss-Bonnet gravity. If so, one would expect an analogous geometric inequality to hold. We answer this in the affirmative within spherical symmetry: 
\begin{theorem}\label{PI:GB}
Consider a spherically symmetric asymptotically flat or asymptotically hyperbolic initial data set $(M,g,k)$ of the Einstein-Gauss-Bonnet equations (with a negative cosmological constant in the latter case) satisfying~\eqref{eq: constraints GB} with matter sources satisfying the dominant energy condition. Let $\tilde\alpha = (n-2)(n-3)\alpha$ where $\alpha$ is the Gauss-Bonnet coupling constant. Then
\begin{equation}\label{eq: PI GB}
      \mathcal{E} \geq \frac{1}{2} \left[ \left(\frac{A}{\omega_{n-1}}\right)^{\frac{n-2}{n-1}} +  \frac{1}{\ell^2}\left(\frac{A}{\omega_{n-1}}\right)^{\frac{n}{n-1}}
 + \tilde\alpha \left(\frac{A}{\omega_{n-1}}\right)^{\frac{n-4}{n-1}} \right],
\end{equation} where $\mathcal{E}$ is the ADM mass in the asymptotically flat case ($\ell \to \infty$) and the Chrusciel-Nagy energy $E_{\mathrm{hyp}}$ in the asymptotically hyperbolic case (with length scale $\ell = (1 - \widetilde\alpha)^{-1/2}$ with $\widetilde \alpha \in (0,1/2)$) respectively, and $A$ is the area of the apparent horizon. Moreover, equality holds if and only if the initial data can be isometrically embedded as a spatial hypersurface in Schwarzschild(-\emph{AdS}$_{n+1})$ Gauss-Bonnet spacetime of mass $\mathcal{E}$ \eqref{GBSch}.
\end{theorem}
\section{The Einstein-Maxwell theory}\label{section EM}

\subsection{Equations of motion}\mbox{}\vspace{.5ex}

\noindent Consider a spacetime $(\mathbf{M}, \mathbf{g})$ of dimension $d = n + 1$ and the Einstein-Maxwell action 

\begin{equation}
S(\mathbf{g},A,\phi)= \frac{1}{16\pi} \int_{\mathbf{M}} \left(\mathbf{R}(\mathbf{g}) + \frac{n(n-1)}{\ell^2}- 2|F|_{\mathbf{g}}^2\right)\; \mathrm{vol}(\mathbf{g})+\int_{\mathbf{M}}\hat{L}_m(\mathbf{g},\phi),
\end{equation}
where $\mathbf{R}(\mathbf{g})$ is the $\mathbf{g}$-Ricci scalar, $F:=\td A$, $|F|_{\mathbf{g}}^2:=\frac{1}{2!}F_{ab}F^{ab}$, $\phi$ denotes the collection of non-charged fields, and $\hat{L}_m$ their Lagrangian. The equations of motion are
\begin{subequations}
\begin{align}
&\mathbf{G}_{ab} =2 \left(F_{ac}F_b^{~c} -\frac{1 }{2}|F|_{\mathbf{g}}^2\mathbf{g}_{ab}\right) + \frac{n(n-1)}{2 \ell^2} \mathbf{g}_{ab}+8\pi\hat{T}_{ab}\label{eq: EOM wrt g},\\
&\td \star_{\mathbf{g}} F = 0\label{eq: EOM wrt A},
\end{align}
\end{subequations}
where $\mathbf{G}_{ab}$ is the $\mathbf{g}$-Einstein tensor, the term in parentheses of \eqref{eq: EOM wrt g} is the electromagnetic energy-momentum tensor $T_{ab}$ and $\hat{T}_{ab}$ is the energy-momentum tensor associated with the non-charged fields. Notice that we have also the trivial condition $\td F =0$ that follows from the definition of $F$.

\subsection{Initial Value Problem}\label{subsection: initial value problem EM}\mbox{}\vspace{.5ex}

\noindent We now work out the constraint equations induced on a spacelike hypersurface $M$, relevant for the initial value formulation. We use $\{i,j,k\ldots\}$ abstract indices on $M$ and denote $\imath:M\hookrightarrow \mathbf{M}$ the inclusion map (hence $\imath_a^i$ denotes the pullback/pushforward) which induces the metric $g:=\iota^*\textbf{g}$ and its Levi-Civita covariant derivative $D$. We also have the future-pointing unit $\mathbf{g}$-normal vector field  $n^a$ over $M$. The orientation of $M$ is given by $\mathrm{vol}(g):=\imath^*\iota_{\vec{n}}\mathrm{vol}(\mathbf{g})$ (equivalently, with a slight abuse of notation, $\mathrm{vol}(\mathbf{g})=-n\wedge \mathrm{vol}(g)$). It will be useful to introduce the following fields:
    \begin{subequations}
\begin{align}
    &\text{Extrinsic curvature} &&  k=\imath^*\nabla n&&  k_{ij}=\imath^a_i\imath^b_j \nabla_a n_b\\
    &\text{Electric field one-form} && E=\imath^*\iota_{\vec{n}}F&& E_{\scriptscriptstyle i}=\imath^b_i n^aF_{ab}\\
    &\text{Magnetic field $(n-2)$-form} && B=\star_g(\imath^*F)&& B_{i_3\cdots i_n}=\frac{1}{2}(\imath^*F)_{ij}\mathrm{vol}(g)^{\phantom{i_3\cdots i_n}ij}_{i_3\cdots i_n}\\
    &\text{Uncharged energy density} && \hat{\mu}=\iota_{\vec{n}}\iota_{\vec{n}}\hat{T}&& \hat{\mu}=n^an^b\hat{T}_{ab}\\
    &\text{Uncharged momentum density} &&  \hat{J}=\imath^*\iota_{\vec{n}}\hat{T}&&  \hat{J}_i=\imath^b_i n^a\hat{T}_{ab}.
\end{align}
    \end{subequations}

Contracting \eqref{eq: EOM wrt g} with $n^an^b$ and $\imath^a_in^b$, and contracting  \eqref{eq: EOM wrt A} with $n^a$, leads to the constraints
\begin{subequations}\label{eq: constraint EM}
\begin{align}
    &16 \pi \hat{\mu} = R(g) + (\text{tr}_g k)^2 - |k|_g^2 - 2|E|_g^2 -2 |B|_g^2 + \frac{n(n-1)}{\ell^2}\label{eq: constraint mu}\\
    &8\pi \hat{J} = \text{div}_g (k - (\text{tr}_g k)g) + 2 \star_g (E \wedge B)\label{eq: constraint J}\\
    &\text{div}_g E=0\qquad\qquad\equiv\qquad\qquad D_i E^i=0\label{eq: constraint E},
\end{align} where $|\omega|^2:=\omega_{a_1\ldots a_p}\omega^{a_1\ldots a_p}/p!$ denotes the norm of a $p${}-{}form $\omega$. Finally, from $\td F=0$, we obtain the last constraint 
\begin{align} 
\text{div}_g B =0\qquad\qquad\equiv\qquad\qquad D_i B^{ij_2\cdots j_{n-3}}=0\label{eq: constraint B}.
\end{align}
    \end{subequations}

The Riemannian manifold $(M,g,k,E,B)$ forms an \emph{Einstein-Maxwell initial data set} if \eqref{eq: constraint mu}-\eqref{eq: constraint B} are satisfied for some $(\hat{\mu}, \hat{J})$.

\subsection{Spherically symmetric initial data set}\label{SSdata} \mbox{}\vspace{.5ex}

\noindent We consider a spherically symmetric initial data set $(M,g,k,E,B)$ satisfying \eqref{eq: constraint EM}. This symmetric condition forces several restrictions on all the elements of the initial data set. First, notice that $M$ is a warped product $I\times_\rho\mathbb{S}^{n-1}$ for some interval $I$ (we assume the initial data set has an apparent horizon boundary, which is invariant under the isometries of the ambient geometry, thus $I= [r_+,\infty)$ for some $r_+>0$) and some positive function $\rho:I\to\mathbb{R}^+$. That means that if we introduce a radial coordinate $s$, the metric can be written as
\begin{equation}\label{eq: spherically symmetric metric}
g = \td s^2 + \rho(s)^2g_{\scriptscriptstyle \mathbb{S}}. 
\end{equation} Notice that the pullback of $g_{\scriptscriptstyle \mathbb{S}}$ to $\mathbb{S}^{n-1}$ (with abstract indices $\{A,B,C\ldots\}$) is the standard round metric $\td\Omega^2$. It is worth mentioning that the vector field $(\partial_s)^i=g^{ij}(\td s)_j$ is tangent to the unit radial geodesics. 

The symmetric condition on the extrinsic curvature implies that there exist functions $\widetilde{k}_{\scriptscriptstyle I},\widetilde{k}_{\scriptscriptstyle \mathbb{S}}:I\to\mathbb{R}$ such that
\begin{equation}\label{eq: spherically symmetric extrinsic curvature}k = \widetilde{k}_{\scriptscriptstyle I}(s) \td s^2 + \frac{\widetilde{k}_{\scriptscriptstyle \mathbb{S}}(s)}{n-1} \rho(s)^2 g_{\scriptscriptstyle \mathbb{S}}. 
\end{equation}
Finally, the symmetric condition implies also that the electric $1$-form field and the magnetic $(n-2)$-form field must be purely radial, so that $E$ is given by
\begin{equation}\label{eq: spherically symmetric electric field}E = E_{\scriptscriptstyle I}(s) \td s 
\end{equation}
for some function $E_{\scriptscriptstyle I}:I\to\mathbb{R}$. Meanwhile, the magnetic field must vanish (unless $n-2=1$, in which case it is a $1$-form, which also has to be radial), as we now argue. We first demonstrate the following lemma. 
\begin{lemma}\label{rotinvforms}
    There are no $SO(n)$-invariant $(n-2)${}-{}forms on $\mathbb{S}^{n-1}$. 
\end{lemma}
\begin{proof}
    Suppose there was a non-vanishing $SO(n)$-invariant $(n-2)${}-{}form $\alpha_{A_1 \ldots A_{n-2}}$. Then 
    \[\beta_{AB} = \alpha_{AC_2\cdots C_{n-2}} \alpha_B^{\phantom{B}C_2 \cdots C_{n-2}} - \frac{|\alpha|_{\scriptscriptstyle \mathbb{S}}^2}{n-1} (\td\Omega^2)_{AB}\]
    is a traceless symmetric $SO(n)$-invariant $(0,2)$ tensor on $\mathbb{S}^{n-1}$ ($\td\Omega^2$ is used to raise/lower the indices and compute the norm). Then, $\gamma_t:= \td\Omega^2+ t \beta$ is a 1-parameter family of $SO(n)$-invariant metrics on $\mathbb{S}^{n-1}$ for sufficiently small $|t|$. However, up to scaling, there is a unique such metric, namely $\td\Omega^2$. Thus, $\beta$ (which is traceless) must be proportional to $\td\Omega^2$ (which is not traceless), and hence it must be zero, leading to
    \begin{equation}\label{eq: alpha^2=Id}\alpha_{AC_2\ldots C_{n-2}} \alpha^{BC_2 \ldots C_{n-2}}= \frac{|\alpha|_{\scriptscriptstyle \mathbb{S}}^2}{n-1}\delta_A^B.
    \end{equation}
    Since $\alpha$ is non-vanishing, being a $(n-2)${}-{}form (``codimension 1''), it must have a nontrivial kernel. Thus, there exists $V^A\neq0$ such that $\alpha_{AC_2\ldots C_{n-2}}V^A=0$. Plugging this into \eqref{eq: alpha^2=Id}, shows that $V^A=0$, leading to a contradiction.    
\end{proof} Let $X_s: \mathbb{S}^{n-1} \hookrightarrow M$ be the embedding mapping of the unit sphere into the sphere of radius $s$ in $M$. Define $a := \iota_{\partial_s}B$ and $b:=B - \td s \wedge a$. We may then write $B = \td s \wedge a + b$. Note that $B$ and $\td s$ are both rotation-invariant, and hence so are $a$ and $b$. Consider the pullback $\tilde b = (X_s)^* b \in \Omega^{n-2}(\mathbb{S}^{n-1})$. By the above Lemma, $(X_s)^*b =0$ for all $s$ and hence $b=0$ since $b$ has no $\mathrm{d}s$ term. Now consider $\mu:=\star_g B=\star_g(\mathrm{d}s\wedge a)\in \Omega^2(M)$. It is clear that $\mu$ has no $\mathrm{d}s$ term, hence $\iota_{\partial_s} \mu=0$. Since $\td \star_g B=0$, $\mu$ must be independent of $s$, so $\mu$ can be understood (with a slight abuse of notation) as a $2$-form on the sphere $\mu \in \Omega^2(\mathbb{S}^{n-1})$. This means that, in particular, $\mu$ is closed (and in fact exact for $n>3$). Since $\mu$ can be locally expressed at a point $p \in \mathbb{S}^{n-1}$ as a square antisymmetric matrix, its rank must be even. Hence if $n-1$ is odd, $\mu_p$ must have a non-trivial kernel, and a minor modification of the argument used in Lemma \ref{rotinvforms} implies $\mu_p =0$. By rotational invariance it must vanish everywhere on the sphere. If $n-1>2$ is even, $\mu$ cannot have maximal rank either, for otherwise it would define a symplectic form (a closed, non-degenerate two form), which does not exist on even-dimensional spheres $\mathbb{S}^{n-1}$ with $n> 3$. Hence $\mu=0$ by Lemma \ref{rotinvforms}. The only other case is $n=3$, in which case $B = B(s) \td s$. Then $E \wedge B=0$. By electromagnetic duality, we may consider a transformation into a duality frame in which $B=0$. In summary, we will focus on the case with $B=0$ in the following. 

%\begin{remark}
 %   If $p=0$, then $\alpha$ is a function on $\mathbb{S}^{n-1}$. To be $SO(n)$-invariant, it must be constant but not necessarily zero. Likewise, if $p=n-1$ and we impose that $\alpha$ is $SO(n)$-invariant, then $\star_{\scriptscriptstyle \mathbb{S}}\alpha$ is a 0-form which must be constant. Hence, $\alpha$ is a multiple of the volume form, which is indeed $SO(n)$-invariant. Notice that in this case, the previous argument does not apply since it is immediate to prove that $\beta_{AB}$ vanishes identically using $\mathrm{vol}(\td \Omega^2)_{AC_2\cdots C_{n-1}}\mathrm{vol}(\td \Omega^2)^{BC_2\cdots C_{n-1}}=(n-2)!\delta^B_A$.

%\end{remark}

%We conclude that $E$ must be radial and $B$ is either zero or radial. In any case, $E \wedge B =0$, so the momentum constraint \eqref{eq: constraint J} is unaffected by the presence of electromagnetic fields.  In $n=3$, by electromagnetic duality, we may consider a transformation into a duality frame in which $B=0$. Hence, we will focus on the case with $B=0$ in the following. 

 Recall that the null expansions are given by $\theta^\pm = k_{\scriptscriptstyle \mathbb{S}} \pm H$, where $\mathbb{S}$ is a level set of $r$, $k_{\scriptscriptstyle \mathbb{S}}$ coincides with the trace of $k$ with respect to the induced metric $\rho(s)^2\td\Omega^2$ on $\mathbb{S}$ (see \eqref{eq: spherically symmetric extrinsic curvature}), and $H$ is the mean curvature of the level set.  The requirement that the level set be a MOTS is $\theta^+=0$. The assumption that the apparent horizon boundary is an outermost MOTS is equivalent to imposing $\rho'(s) > 0$. Furthermore, the outermost property implies that the outermost MOTS is also outer minimizing \cite[Theorem 7.46]{Leetext}. We may thus use $r:=\rho(s)$ as a coordinate. Denoting $\widetilde{V}(s):=\rho'(s)^2$ and $V(r):=\widetilde{V}(\rho^{-1}(r))$, the metric \eqref{eq: spherically symmetric metric} and extrinsic curvature \eqref{eq: spherically symmetric extrinsic curvature} read
\begin{equation}\label{sphericalmetric}
g = \frac{\td r^2}{V(r)} + r^2 g_{\scriptscriptstyle \mathbb{S}}\qquad\qquad k = k_{\scriptscriptstyle I}(r) \frac{\td r^2}{V(r)} + \frac{k_{\scriptscriptstyle \mathbb{S}}(r)}{n-1} r^2 g_{\scriptscriptstyle \mathbb{S}}, 
\end{equation}where $k_{\scriptscriptstyle I}(r)=\widetilde{k}_{\scriptscriptstyle I}(\rho^{-1}(r))$ and $k_{\scriptscriptstyle \mathbb{S}}(r):=\widetilde{k}_{\scriptscriptstyle \mathbb{S}}(\rho^{-1}(r))$. Observe that the outward unit normal to a surface of constant $r$ is given by $\nu^i=\sqrt{V}\partial_r^i$, which implies that $\nu_i=(\td r)_i/\sqrt{V}$. In particular, its mean curvature is given by
\begin{equation}\label{eq: H^2}
H =(n-1)\frac{ \sqrt{V(r)}}{r}.
\end{equation}

\subsection{Definition of charged mass and its properties}\mbox{}\vspace{.5ex}

\noindent We generalize the elementary proof of the spacetime Penrose inequality in spherical symmetry outlined by \cite{Leetext}. The approach is to define a quasi-local mass (the Hawking mass) associated with spheres (level sets of an appropriate radial function defined on the initial data set) and then show that this quantity is non-decreasing along a radial flow and approaches the total energy in the designated asymptotic region. Motivated by Disconzi-Khuri \cite{DK}, we modify the Hawking mass by adding a term involving a cosmological constant and the total charge of the system (see Equation \eqref{eq: charge}).  

\begin{definition}[Charged Hawking (Misner-Sharp) mass]
\begin{equation}\label{CHM}
\begin{array}{l}
m_{\mathrm{H}}(r) := \dfrac{r^{n-2}}{2}\left[ 1  + \dfrac{r^2}{(n-1)^2} \left(k_{\scriptscriptstyle \mathbb{S}}^2 - H^2 \right) \right]\\[2.5ex]
m_q(r) := \dfrac{q^2}{2r^{n-2}}\\[2.5ex]
m_{\ell}(r) := \dfrac{r^{n}}{2\ell^2}
\end{array}\quad\longrightarrow\quad m_{\text{CH}}(r) :=m_{\mathrm{H}}(r)+m_{\ell}(r)  +m_q(r)
\end{equation}   
\end{definition}  Observe that $k_{\scriptscriptstyle \mathbb{S}}^2 - H^2 = \theta^+ \theta^-$, the product of the outward and inward expansions. For a MOTS, this term vanishes. Hence, the quasi-local mass can be rewritten in a standard form as an integral over a surface $S$. 

As mentioned before, once we have defined the Charged Hawking mass, we must prove two things. First, that it is increasing, and second, that it tends to the total energy of the space-time (ADM energy or AdS energy) in the asymptotic limit $r\to\infty$.

\subsubsection{Monotonicity}
\begin{proposition}\label{lemma: mass mono} If the constraint equations \eqref{eq: constraint EM} are satisfied and the system is spherically symmetric \eqref{sphericalmetric}, the charged Hawking mass is non-decreasing as a function of $r$.
\end{proposition}
\begin{proof}
 First, notice that using \eqref{eq: H^2}, we can rewrite the Hawking mass as
 \begin{equation}\label{eq: m_ch -q^2=F}
m_{\mathrm{H}}(r)= \frac{r^{n-2}}{2}\left[ 1 -V(r)+ r^2\left(\frac{k_{\scriptscriptstyle \mathbb{S}}}{n-1}\right)^{\!\!2} \right].
\end{equation}
It is easy to check that the scalar curvature of the spherically symmetric metric \eqref{sphericalmetric} is 
\begin{equation}\label{scalg}
R(g)  = \frac{n-1}{r^{n-1}} \frac{\td }{\td r} \left[ r^{n-2} (1-V(r)) \right].
\end{equation}
Thus,
\begin{equation}\label{eq: m'H+m'l}\begin{aligned}
m_{\mathrm{H}}&'(r)+m_{\ell}'(r)= \frac{1}{2}\frac{r^{n-1} }{n-1}R(g) +\frac{r^{n-1}}{2} \left( n \left(\frac{k_{\scriptscriptstyle \mathbb{S}}}{n-1}\right)^{\!2}  + 2 r \frac{k_{\scriptscriptstyle \mathbb{S}}}{n-1} \frac{k'_{\scriptscriptstyle \mathbb{S}}}{n-1}\right)+n\frac{r^{n-1}}{2\ell^2}\\
&=\frac{r^{n-1} }{2(n-1)}\left[R(g) -  \frac{n}{n-1}k_{\scriptscriptstyle \mathbb{S}}^2+\frac{n(n-1)}{\ell^2}\right]+r^{n-1}\frac{k_{\scriptscriptstyle \mathbb{S}}}{n-1} \left( n \frac{k_{\scriptscriptstyle \mathbb{S}}}{n-1}+ r  \frac{k'_{\scriptscriptstyle \mathbb{S}}}{n-1}\right)\\
&=\frac{r^{n-1} }{2(n-1)}\left\{R(g) +(\text{Tr}_g k)^2-|k|^2_g+\frac{n(n-1)}{\ell^2}-2k_{\scriptscriptstyle \mathbb{S}}\left(\text{Tr}_g k-\frac{nk_{\scriptscriptstyle \mathbb{S}}}{n-1}- r  \frac{k'_{\scriptscriptstyle \mathbb{S}}}{n-1}\right)\right\}\\
&=\frac{r^{n-1} }{2(n-1)}\left\{R(g) +(\text{Tr}_g k)^2-|k|^2_g+\frac{n(n-1)}{\ell^2}-2\frac{k_{\scriptscriptstyle \mathbb{S}}}{H}(\nu\cdot{}\text{div}_g (k - (\text{tr}_g k)g))\right\}.
    \end{aligned}
\end{equation}
From the first to the second line, we have used that in spherical symmetry, $\text{Tr}_g k = k_{\scriptscriptstyle I} + k_{\scriptscriptstyle \mathbb{S}}$ and $|k|_g^2 = k_{\scriptscriptstyle I}^2 +  k_{\scriptscriptstyle \mathbb{S}}^2/(n-1)$. Meanwhile, the last line is obtained by contracting the RHS of \eqref{eq: constraint J} (recall that we are assuming $\star_g(E\wedge B)=0$) with the normal $\nu_i =(\td r)_i/\sqrt{V}$. Assuming the constraints \eqref{eq: constraint EM}, we have
\begin{equation}\label{eq: m'CH}
m_{\mathrm{CH}}'(r)=m_{\mathrm{H}}'(r)+m_{\ell}'(r)+m_{q}'(r)=\frac{8\pi r^{n-1} }{n-1}\left[\hat{\mu}-\frac{k_{\scriptscriptstyle \mathbb{S}}}{H}(\hat{J}\cdot{}\nu)\right]+\frac{r^{n-1} }{n-1}|E|_g^2-\frac{n-2}{2r^{n-1}} q^2.
\end{equation}

We now focus on the last term:
\begin{align*}
     m'_q(r)&=-\frac{n-2}{2r^{n-1}} q^2\overset{\eqref{eq: charge}}{=}-\frac{n-2}{2r^{n-1}}\frac{2}{(n-1)(n-2)} \frac{1}{\omega_{n-1}^2}\left( \lim_{R \to \infty} \int_{S_R} E(\nu)\; \mathrm{vol}(S_R^{n-1})\right)^{\!2}\\
     &=-\frac{1}{(n-1)r^{n-1}\omega_{n-1}^2} \lim_{R \to \infty} \left(\int_{V_{r\!R}}\text{div}_gE\; \mathrm{vol}(V_{r\!R})-\int_{S_r} E(\nu)\; \mathrm{vol}(S_r^{n-1})\right)^{\!2}\geq\\
     &\geq-\frac{r^{n-1}}{(n-1)(r^{n-1}\omega_{n-1})^2} \lim_{R \to \infty} \left(\int_{S_r} E(\nu)^2\; \mathrm{vol}(S_r^{n-1})\right)\left(\int_{S_r}1^2\; \mathrm{vol}(S_r^{n-1})\right)\\
     &=-\frac{r^{n-1}}{(n-1)r^{n-1}\omega_{n-1}} \int_{S_r} |E|_g^2\; \mathrm{vol}(S_r^{n-1})=-\frac{r^{n-1}|E|_g^2}{n-1},
\end{align*}
 where $V_{r\!R}$ denotes the space in between the spheres $S_r$ and $S_R$. In the second line, we have used the divergence theorem; in the third, the Cauchy-Schwarz inequality for integrals together with the constraint $\text{div}_g E =0$, and in the last two equalities, we have used that $E$ is purely radial. 

Gathering the previous two equations, we have
\begin{equation}\label{eq: m'>0}
    m_{\text{CH}}'(r)\geq\frac{8\pi r^{n-1} }{n-1}\left[\hat{\mu}-\frac{k_{\scriptscriptstyle \mathbb{S}}}{H}(\hat{J}\cdot{}\nu)\right]\geq0.
\end{equation}
The last inequality follows from the fact that, as argued in \cite{Leetext}, the outermost condition implies $H > k_{\scriptscriptstyle \mathbb{S}}$ in the interior of $M = [r_+, \infty) \times \mathbb{S}^{n-1}$ (otherwise, at some point $H = k_{\scriptscriptstyle \mathbb{S}}$ by continuity, implying the existence of another MOTS outside the boundary). Then the dominant energy condition $\hat{\mu}\geq |\hat{J}| \geq |(\hat{J} \cdot \nu)|$ implies $m_{\text{CH}}'(r) \geq 0$, with equality if and only if $\hat{\mu} = |\hat{J}| =0$ and $|E| = 0$.  \end{proof}

\subsubsection{Asymptotic limit} \mbox{}\vspace{.5ex}

\noindent We wish to show the charged Hawking mass coincides with the total energy in the limit $r \to \infty$, i.e., with $E_{\mathrm{ADM}}$ \eqref{eq: EADM} (notice that a spherically symmetric initial data set can be shown to have vanishing linear momenta; thus the ADM energy coincides with the ADM mass) or $E_{\mathrm{hyp}}$ \eqref{eq: Ehyp}. This result is quite standard, but we give the argument for convenience. 
\begin{lemma}\label{lemma: lim m_CH} \begin{equation}
    \lim_{r \to \infty}m_{\mathrm{CH}}(r) =  \begin{cases} E_{\mathrm{ADM}} & \text{asymptotically flat } (\ell = \infty) \\ E_{\mathrm{hyp}} & \text{asymptotically hyperbolic } (\ell = 1) \end{cases}
    \end{equation}
\end{lemma} \begin{proof} 
Since the charged term $m_q$ decreases to zero when $r\to\infty$, we can ignore that term in the following.

Consider first the asymptotic flatness case, which implies $m_\ell=0$ and $\text{Tr}_g k = O_1(r^{-\kappa})$ for some $\kappa > (n-2)/2$. The latter implies $k_{\scriptscriptstyle \mathbb{S}}= O_1(r^{-\kappa-2})$, so $r^n k_{\scriptscriptstyle \mathbb{S}}^2 = O_1(r^{n - 2\kappa - 4}) = O_1 (r^{-2 - \epsilon})$ for some $\epsilon > 0$. We then have that the limit of \eqref{eq: m_ch -q^2=F} is
\begin{equation}\label{eq: lim m_H} 
    \lim_{r\to\infty}m_{\mathrm{H}}(r)=\lim_{r\to\infty} \frac{r^{n-2}}{2}( 1 -V(r)).
\end{equation}

For the spherically symmetric metric \eqref{sphericalmetric}, the fall-off relative to the Euclidean metric is
\begin{equation}
    g - \delta = \left(\frac{1}{V}-1\right) \td r^2 = \frac{1-V}{V} \frac{x_i x_j}{r^2} \td x^i \td x^j,
\end{equation} where $\{x^i \}$ is a standard Cartesian chart with $r = \sqrt{x^i x^j \delta_{ij}}$. Substituting this into the ADM energy formula, one finds 
\begin{equation}
    E_{\mathrm{ADM}} = \lim_{r \to \infty} \frac{r^{n-1}}{2} \frac{1-V(r)}{rV(r)}  \overset{\eqref{eq: lim m_H}}{=} \lim_{r\to \infty} m_{\mathrm{H}}(r),
\end{equation}  
where in the last equality, we also used $V(r)\to1$ when $r\to\infty$. 

Consider now the asymptotically hyperbolic case (we fix $\ell =1$ without loss of generality). From equation \eqref{eq: hatg}, observe that $V_{\mathrm{hyp}}(r)/V(r) \to 1$ and that the term $r^n k_{\scriptscriptstyle \mathbb{S}}^2 = O_1(r^{n - 2\kappa - 4}) = O_1 (r^{-2 - \epsilon})$ can be neglected in the limit $r\to\infty$. Thus,
\begin{equation}\label{eq: lim m_H hyperbolic} 
    \lim_{r\to\infty}m_{\mathrm{CH}}(r)=\lim_{r\to\infty} \frac{r^{n-2}}{2}\left( 1 +\frac{r^2}{\ell^2}-V(r)\right).
\end{equation}

For a spherically symmetric metric, the deviation tensor \eqref{eq: hatg} is given by
\begin{equation}
    \hat{g}:=g-b= \frac{V_{\mathrm{hyp}}-V}{V_{\mathrm{hyp}}V}\td r^2= \frac{V_{\mathrm{hyp}}-V}{V}\mathcal{V}\otimes\mathcal{V},
\end{equation} where $\mathcal{V}_i=(\td r)_i/\sqrt{V_{\mathrm{hyp}}}$.

A computation shows 
\begin{equation}
    E_{\mathrm{hyp}} = \lim_{r \to \infty} \frac{V_{\mathrm{hyp}}(r)}{V(r)}\frac{ r^{n-2}}{2}(V_{\mathrm{hyp}}(r) -V(r)) =\lim_{r\to\infty}m_{\mathrm{CH}}(r),
\end{equation} where we are using the asymptotic behaviour of $V(r)$ as $r \to \infty$.
\end{proof}

\subsection{Penrose inequality}\mbox{}\vspace{.5ex}

\noindent In this section, we prove Theorem \ref{Thm:PI}.
\begin{proof}
Let $r_+$ correspond to the location of the apparent horizon. Then
\begin{equation}
     \mathcal{E}\overset{\text{lemma}}{\underset{\ref{lemma: lim m_CH}}{=}}\lim_{r\to\infty}m_{\mathrm{CH}}(r)\overset{\text{lemma}}{\underset{\ref{lemma: mass mono}}{\geq}}m_\mathrm{CH}(r_+) = \frac{r_+^{n-2}}{2}+\frac{r^n_+}{2\ell^2}+\frac{q^2}{2r_+^{n-2}},
\end{equation}
where in the last equality we have used the definition of $m_{\mathrm{CH}}$ and the equality $k_{\scriptscriptstyle \mathbb{S}}^2 - H^2 = \theta^+ \theta^-$, which vanishes for a MOTS. The Penrose inequality \eqref{eq: PenroseInequality} follows from the area being given by $A=r_+^{n-1}\omega_{n-1}$.

Now it remains to prove the rigidity statement. Namely, if the equality saturates, the initial data set must arise from a spatial hypersurface in the static Reissner–Nordström(-AdS$_{n+1})$ spacetime. Let us recall the definition and the required properties of this space-time. Its metric can be written as
\begin{equation}\label{eq: RN metric}
\mathbf{g}_{\scriptscriptstyle\mathrm{RN}} =  -V_{\scriptscriptstyle\mathrm{RN}}(r) \td t^2 + \frac{\td r^2}{V_{\scriptscriptstyle\mathrm{RN}}(r)} + r^2g_{\scriptscriptstyle \mathbb{S}}, \qquad V_{\scriptscriptstyle\mathrm{RN}}(r)= 1 - \frac{2M}{r^{n-2}} + \frac{Q^2}{r^{2n-4}} + \frac{r^2}{\ell^2};
\end{equation}
where $r$ ranges in $(r_0,\infty)$, the unbounded interval where $V_{\scriptscriptstyle\mathrm{RN}}(r)>0$.

The initial value problem to be considered is given on a spherically symmetric hypersurface, i.e. given as the level set $M_{\scriptscriptstyle (f)}=\{t = f(r)\}$ for some $f$ satisfying $|V_{\scriptscriptstyle\mathrm{RN}} f'|<1$ ($M_{\scriptscriptstyle (f)}$ must be space-like). In the following, we denote $s$ the sign of $f'$. 

\begin{lemma}\label{lemma: initial data RN}
    The initial data on $M_{\scriptscriptstyle (f)}$ induced by \eqref{eq: RN metric} satisfying the constraint equations \eqref{eq: constraint EM} is:
\begin{equation}\label{eq: RN initial data}
      \begin{alignedat}{3}
        &g^{\scriptscriptstyle (f)} = \frac{\td r^2}{V_{\!\scriptscriptstyle f}} + r^2g_{\scriptscriptstyle \mathbb{S}} &\qquad\qquad&\hat{\mu}^{\scriptscriptstyle (f)}=0 &\qquad\qquad&  E^{\scriptscriptstyle (f)} =E^{\scriptscriptstyle (f)}_{\scriptscriptstyle I} \td r\\
        &  k^{\scriptscriptstyle (f)} = k_{\scriptscriptstyle I}^{\scriptscriptstyle(f)} \frac{\td r^2}{V_{\!\scriptscriptstyle f}} + \frac{k_{\scriptscriptstyle \mathbb{S}}^{\scriptscriptstyle(f)}}{n-1} r^2 g_{\scriptscriptstyle \mathbb{S}} &\qquad\qquad&\hat{J}^{\scriptscriptstyle (f)}=0  &\qquad\qquad&B^{\scriptscriptstyle (f)}=0,
        \end{alignedat}
        \end{equation}
    where     
    \begin{equation}\label{eq: def Vf, k^f E_i}
      \begin{alignedat}{3}&V_{\!\scriptscriptstyle f}:= \frac{V_{\scriptscriptstyle\mathrm{RN}}}{1 - (V_{\scriptscriptstyle\mathrm{RN}} f')^2}\qquad k_{\scriptscriptstyle I}^{(f)} =\frac{\td}{\td r} \left(r\frac{k_{\scriptscriptstyle \mathbb{S}}^{\scriptscriptstyle(f)}}{n-1}\right)  \qquad k_{\scriptscriptstyle \mathbb{S}}^{(f)} = -s\frac{n-1}{r} \sqrt{V_{\!\scriptscriptstyle f} - V_{\scriptscriptstyle\mathrm{RN}}}\\
    &E^{\scriptscriptstyle (f)}_{\scriptscriptstyle I}=\sqrt{\frac{(n-1)(n-2)}{2V_f}}\frac{Q}{r^{n-1}}.
        \end{alignedat}
        \end{equation}
    Moreover, the parameters $(M,Q)$ of $V_{\scriptscriptstyle\mathrm{RN}}$ correspond to the ADM/hyperbolic mass and the charge.  
\end{lemma}
\begin{proof}
     We introduce the new time coordinate $\tau = t - f(r)$. It is then a short computation to check
\begin{equation}
     \mathbf{g}_{\scriptscriptstyle\mathrm{RN}}  = -V_{\!\scriptscriptstyle f} \td \tau^2 + \frac{1}{V_{\!\scriptscriptstyle f}} \left( \td r - V_{\!\scriptscriptstyle f}V_{\scriptscriptstyle\mathrm{RN}}  f' \td \tau\right)^2 + r^2g_{\scriptscriptstyle \mathbb{S}}.
\end{equation} Thus, the induced metric on $M_{\scriptscriptstyle (f)}=\{\tau=0\}$ is simply
\begin{equation}
    g^{\scriptscriptstyle (f)} = \frac{\td r^2}{V_{\!\scriptscriptstyle f}} + r^2g_{\scriptscriptstyle \mathbb{S}},
\end{equation} with lapse and shift
\begin{equation}
    \textbf{N} =  \sqrt{V_{\!\scriptscriptstyle f}}, \qquad \mathcal{N}_i = - V_{\scriptscriptstyle\mathrm{RN}} f' (\td r)_i.
\end{equation}
The extrinsic curvature can be obtained from the lapse and the shift through the formula
\begin{equation}
    k^{\scriptscriptstyle (f)}_{ij} = \frac{1}{2N} \left[ D_i N_j + D_j N_i \right],
\end{equation}
where $D$ denotes the $g^{\scriptscriptstyle (f)}$-Levi-Civita connection. It is not hard to check that the non-zero components are  
\begin{equation}\label{Kf}
    \frac{k_{\scriptscriptstyle \mathbb{S}}^{\scriptscriptstyle(f)}}{n-1} = -\frac{s}{r}\sqrt{V_{\!\scriptscriptstyle f} - V_{\scriptscriptstyle\mathrm{RN}}}\qquad \qquad k^{\scriptscriptstyle (f)}_{\scriptscriptstyle I}:=\nu^i\nu^jk^{\scriptscriptstyle (f)}_{ij}  =\frac{\td}{\td r} \left(r\frac{k_{\scriptscriptstyle \mathbb{S}}^{\scriptscriptstyle(f)}}{n-1}\right), 
\end{equation}
where $\nu^i = \sqrt{V_{\!\scriptscriptstyle f}} \partial_r^i$. For that computation, it is useful to realize
\begin{equation}
    V_{\scriptscriptstyle\mathrm{RN}} f'=s\sqrt{1-\frac{V_{\scriptscriptstyle\mathrm{RN}}}{V_{\!\scriptscriptstyle f}}}.
\end{equation} 

Since the Reissner–Nordström solution has no uncharged matter, $\hat{\mu}^{\scriptscriptstyle (f)}=0$ and $\hat{J}^{\scriptscriptstyle (f)}=0$. Recall, because of the spherical symmetry, $B^{\scriptscriptstyle (f)}=0$ and the electric field is radial $E=E_{\scriptscriptstyle I}\td r$. The term $E_{\scriptscriptstyle I}$ can be easily obtained using the constraint equation \eqref{eq: constraint mu}
\begin{equation}\label{EfieldRN}
2V_{\!\scriptscriptstyle f}(E^{\scriptscriptstyle (f)}_{\scriptscriptstyle I})^2=2|E^{\scriptscriptstyle (f)}|_{g^{\scriptscriptstyle(f)}}^2=(n-1)(n-2)\left(\frac{Q}{r^{n-1}}\right)^{\!2}, 
\end{equation}
where we have used $\text{Tr}_g k^{\scriptscriptstyle(f)} = k^{\scriptscriptstyle(f)}_{\scriptscriptstyle I} + k^{\scriptscriptstyle(f)}_{\scriptscriptstyle \mathbb{S}}$, $|k^{\scriptscriptstyle(f)}|_g^2 = (k^{\scriptscriptstyle(f)}_{\scriptscriptstyle I})^2 +  (k^{\scriptscriptstyle(f)}_{\scriptscriptstyle \mathbb{S}})^2/(n-1)$ and Equations \eqref{Kf},  Equation \eqref{scalg} with $V_{\!\scriptscriptstyle f}$ instead of $V$, and the definition of $V_{\scriptscriptstyle\mathrm{RN}}$ given in \eqref{eq: RN metric}. Therefore, we indeed recover the initial data given in Equation \eqref{eq: RN initial data}.

It only remains to prove $M$ is the ADM/hyperbolic mass and $Q$ is the charge. The former is immediate using lemma \ref{eq: lim m_H hyperbolic} and the limits given in \eqref{eq: lim m_H} or \eqref{lemma: lim m_CH}. For the charge we use the definition given in \eqref{eq: charge}, the electric field we computed in \eqref{EfieldRN}, and the fact $E(\nu)=E_{\scriptscriptstyle I}\sqrt{V_{\!\scriptscriptstyle f}}$:
\begin{equation}
    q =\frac{1}{\omega_{n-1}} \lim_{r \to \infty} \int_{S_r}\frac{Q}{r^{n-1}} \mathrm{vol}(S_r^{n-1})=Q.
\end{equation}
\end{proof}

The goal now is to show, for a given initial data set $(g,k,E,B)$ that saturates the charged spherically symmetric Penrose inequality, there is some function $f(r)$ such that $g=g^{\scriptscriptstyle (f)}$ and $k=k^{\scriptscriptstyle (f)}$.

Recall that $(g,k)$ are spherically symmetric, hence they are given by Equation \eqref{sphericalmetric} for some $V$, $k_{\scriptscriptstyle I}$ and $k_{\scriptscriptstyle \mathbb{S}}$. It is clear that in order to have $g=g^{\scriptscriptstyle (f)}$, we need to have $V=V_{\!\scriptscriptstyle f}$ for some $f$.

Suppose now equality holds in the Penrose inequality. This implies 
\begin{equation} \mathcal{E} = m_{\mathrm{CH}}(r) = \frac{Q^2}{2r^{n-2}} + \frac{r^{n-2}}{2} \left[ 1 - V(r) + \frac{r^2}{\ell^2} + \left(r\frac{k_{\scriptscriptstyle \mathbb{S}}}{n-1} \right)^{\!2}\right]
\end{equation}
for every $r\geq r_+$ (of course, this means the RHS is constant). It is immediate to solve for $V$ and find
\begin{equation}\label{eq: V=VRN+KS^2}
    V(r)= V_{\scriptscriptstyle\mathrm{RN}}(r) +\left(r\frac{k_{\scriptscriptstyle \mathbb{S}}}{n-1} \right)^{\!2}+\frac{2}{r^{n-2}}(M-\mathcal{E}).
\end{equation}
Notice that the mass $M$ of the RN slice $M_{\scriptscriptstyle (f)}$ must be equal to the mass $\mathcal{E}$ of the original slice $M$. Moreover, we have seen that $V=V_{\!\scriptscriptstyle f}$. Plugging these expressions back in \eqref{eq: V=VRN+KS^2} and solving for $f'$ leads to
\begin{equation}\label{eq: sol fprime}
    f'(r)^2 = \frac{1}{V_{\scriptscriptstyle\mathrm{RN}}^2}\frac{ \left(r\dfrac{k_{\scriptscriptstyle \mathbb{S}}}{n-1}\right)^{\!2}}{\left(r\dfrac{k_{\scriptscriptstyle \mathbb{S}}}{n-1}\right)^{\!2}+V_{\scriptscriptstyle\mathrm{RN}}},
\end{equation}
which satisfies the condition $|f'(r)V_{\scriptscriptstyle\mathrm{RN}}| <1$, since $V_{\scriptscriptstyle\mathrm{RN}}$ is positive in the interval $(r_0,\infty)$. For this choice of $f$ (defined up to an overall constant and up to a sign on each interval of the support of $f$), we obtain from lemma \ref{lemma: initial data RN} the value of the spherical part of the extrinsic curvature:
\begin{equation}\label{eq: k^f=k}
    k_{\scriptscriptstyle \mathbb{S}}^{(f)} = -s\frac{n-1}{r} \sqrt{V_{\!\scriptscriptstyle f} - V_{\scriptscriptstyle\mathrm{RN}}}\overset{\eqref{eq: V=VRN+KS^2}}{=}-s|k_{\scriptscriptstyle \mathbb{S}}|.
\end{equation}
We can choose the sign of $f$ to get $k_{\scriptscriptstyle \mathbb{S}}^{(f)}=k_{\scriptscriptstyle \mathbb{S}}$. To obtain the radial part of the extrinsic curvature, recall that contracting the RHS of \eqref{eq: constraint J} with $\nu_i =(\td r)_i/\sqrt{V}$ (as we used to derive Equation \eqref{eq: m'H+m'l}) leads to
\begin{equation}
    H(\hat{J}\cdot \nu)=k_{\scriptscriptstyle I}-\frac{k_{\scriptscriptstyle \mathbb{S}}}{n-1}- r  \frac{k'_{\scriptscriptstyle \mathbb{S}}}{n-1}\quad\overset{\eqref{eq: k^f=k}}{\longrightarrow}\quad k_{\scriptscriptstyle I} = \frac{\td}{\td r} \!\left(\!r\frac{k_{\scriptscriptstyle \mathbb{S}}^{(f)}}{n-1}\!\right)+H(\hat{J}\cdot \nu)\overset{\eqref{eq: def Vf, k^f E_i}}{=}k_{\scriptscriptstyle I}^{(f)}+H(\hat{J}\cdot \nu).
\end{equation}
Let us now prove $(\hat{J},\hat{\mu})=(\hat{J}^{\scriptscriptstyle (f)},\hat{\mu}^{\scriptscriptstyle (f)})=0$, which incidentally would also prove $k_{\scriptscriptstyle I}=k_{\scriptscriptstyle I}^{(f)}$. For that, we use $m_{\mathrm{CH}}'=0$. From \eqref{eq: m'CH} we get
\begin{equation}
|E|^2_g =\frac{(n-1)(n-2)}{2}\frac{q^2}{r^{2(n-1)}}, \qquad\qquad  \hat{\mu} = \frac{k_{\scriptscriptstyle \mathbb{S}}}{H}(\hat{J} \cdot \nu).
\end{equation}
The first condition implies, following the results of lemma \ref{lemma: initial data RN}, the electric field $E = (E\cdot \nu) \nu$ of the initial data set is that of a Reissner–Nordström-AdS black hole with charge $q = Q$. The second equation implies that both $\hat{\mu}$ and $\hat{J}\cdot \nu$ vanish (and hence so does $\hat{J}$, since it must be radial). Indeed, we have seen already that $H > k_{\scriptscriptstyle \mathbb{S}}$ for all $r> r_+$, but the dominant energy condition forces $\hat{\mu} \geq |\hat{J}|$. 

Therefore, we have shown that the inequality is saturated if and only if the initial data set can be isometrically embedded as a spatial hypersurface in the Reissner–Nordström(-AdS) spacetime.
\end{proof}

\section{The Gauss-Bonnet gravity theory}
 \noindent In this section, we prove Theorem \ref{PI:GB}, the spacetime Penrose inequality in spherical symmetry within the setting of Einstein-Gauss-Bonnet gravity (see the review \cite{Fernandes:2022zrq}). 
\subsection{Equations of motion}\mbox{}\vspace{.5ex}

\noindent Consider the action functional in spacetime dimension $n+1$:
\begin{equation}
    S(\mathbf{g},A,\psi) = \int \left(\mathbf{R} + \frac{n(n-1)}{\ell^2} + \alpha L_{GB} \right) \mathrm{vol}(\mathbf{g}) + \int_{\mathbf{M}}\widetilde{L}(\mathbf{g},\psi),
\end{equation}
where  $\alpha >0$ is a constant coupling parameter, $\psi$ denotes the collection of additional fields (charged and non-charged), and $\widetilde{L}$ their Lagrangian. The equations of motion are
\begin{equation}
    \mathbf{G}_{ab} + \alpha \mathbf{H}_{ab} = \frac{n(n-1)}{2\ell^2} \mathbf{g}_{ab}+8\pi\widetilde{T}_{ab},
    \end{equation}where $\widetilde{T}_{ab}$ is the energy-momentum tensor associated with the non-metric fields, and we have the geometric `source' term
    \begin{equation}
        \mathbf{H}_{ab}:=2 \left[ \mathbf{R} \mathbf{R}_{ab} - 2 \mathbf{R}_{ac}\mathbf{R}^{c}_{~b} - 2 \mathbf{R}^{cd} \mathbf{R}_{acbd} + \mathbf{R}_a^{~cde}\mathbf{R}_{bcde} \right] - \frac{1}{2} \mathbf{g}_{ab} L_{GB}. 
    \end{equation}

\subsection{Initial Value Problem}\mbox{}\vspace{.5ex}

\noindent Consider an initial data set $(M,g,k)$ with the conventions and definitions of section \ref{subsection: initial value problem EM}. Contracting the field equation with $n^an^b$ and $\imath^a_in^b$ (and after a long computation), we obtain the constraints \cite{Torii:2008ru}:
\begin{subequations}\label{eq: constraints GB}
\begin{align}
    &16\pi \widetilde{\mu}=M+\frac{n(n-1)}{\ell^2}  + \alpha (M^2 - 4M_{ij}M^{ij} + M^{ijkl}M_{ijkl})\label{GB:Hcons}\\ 
      &8 \pi \widetilde{J}_j= N_j + 2 \alpha (M N_j - 2M_j^{~i} N_i + 2 M^{ik}N_{jik} - M_j^{~lik} N_{ikl})\label{GB:Mcons},
\end{align}
\end{subequations}
where $M_{ijkl}:=(\imath^*\mathbf{R})_{ijkl}=\imath^a_i\imath^b_j\imath^c_k\imath^d_l\mathbf{R}_{abcd}$ is the pullback of the $\mathbf{g}$-Riemann tensor (as usual, we define $M_{jl}:=g^{ik}M_{ijkl}$ and $M=g^{jl}M_{jl}$), and $N_{ijk}:=\imath^a_i \imath^b_j\imath^c_kn^d\mathbf{R}_{abcd}$ (we also define $N_j:=g^{ik}N_{ijk}$). In the previous computation, we have used the Gauss-Codazzi and the Codazzi-Mainardi equations:
\begin{subequations}\label{eq: Gauss-Codazzi-Mainardi}
\begin{align}
        &M_{ijkl}= R_{ijkl} + k_{ik} k_{jl} - k_{il} k_{jk}\\
        &N_{ijk}=D_ik_{jk}-D_jk_{ik},
    \end{align}
\end{subequations}  where $R^i_{~jkl}$ is the $g$-Riemann tensor. Notice 
\begin{subequations}\label{eq: M and N}
\begin{align}
 &M=  R(g) + (\text{Tr}_g k)^2 - k_{ij}k^{ij}\\
 &N_j=D^i k_{ij} - D_j (\tr_g k),
\end{align} 
\end{subequations} so when $\alpha =0$, \eqref{eq: constraints GB} reduces to the usual constraints of general relativity.

It is a straightforward exercise to include a Maxwell field using the results from the previous section. We will restrict attention to the Gauss-Bonnet term here. Choquet-Bruhat has addressed the well-posedness of the Cauchy problem associated with these constraints \cite{CB}. 

\subsection{Spherically symmetric initial data set}\mbox{}\vspace{.5ex}

\noindent We consider spherically symmetric data $(M,g,k)$ with metric $g$ and extrinsic curvature $k$ as in \eqref{sphericalmetric}.

\subsection{Definition of Gauss-Bonnet quasi-local mass and its properties}\mbox{}\vspace{.5ex}

\noindent Let's first introduce a quasi-local mass analogous to the charged Hawking mass, adapted for such data satisfying the constraint equations \eqref{GB:Hcons} and \eqref{GB:Mcons}: 
\begin{definition}[Generalized Gauss-Bonnet Hawking mass]
\begin{equation}
m_{\mathrm{GB}}(r):= m_{\mathrm{H}}(r)+m_{\ell}(r)+2\widetilde\alpha\frac{m_{\mathrm{H}}(r)^2 }{r^n} 
\end{equation}
\end{definition} 
We recall that $\widetilde\alpha:= (n-2)(n-3) \alpha$, which vanishes for $n=3$ ($d=4$), in which case $m_{\mathrm{GB}}$ equals the Hawking mass. This is consistent with the GB term being topological. A related definition of a quasi-local Misner-Sharp mass valid for spherically symmetric spacetimes was presented in \cite{Maeda:2007uu}.

\subsubsection{Monotonicity}
\begin{proposition}
    For spherically symmetric initial data satisfying the Gauss-Bonnet constraints, the generalized Gauss-Bonnet quasi-local mass is a monotonically non-decreasing function of $r$. 
\end{proposition}
\begin{proof} We start by computing the $\alpha$ factors of \eqref{eq: constraints GB}, which we denote:
\begin{subequations}\label{eq: constraint GB}
\begin{align}
    &\widetilde{\mu}^{(\alpha)}:=M^2 - 4M_{ij}M^{ij} + M^{ijkl}M_{ijkl}\label{eq: mu_alpha}\\
    &\widetilde{J}^{(\alpha)}_i:=M N_j - 2M_j^{~i} N_i + 2 M^{ik}N_{jik} - M_j^{~lik} N_{ikl},
\end{align}
\end{subequations}
in a spherical symmetric setting \eqref{sphericalmetric}. For that, we break the indices into the normal and tangent parts using the decomposition $g^{ij}=\nu^i\nu^j+X^i_AX^j_B(\td \Omega^2)^{AB}/r^2$ where $\nu^i=\sqrt{V}\partial_r^i$ and $X^i_A$ denotes the pullback/pushforward of the embedding $X:\mathbb{S}^{n-1}\hookrightarrow M$. We compute the following terms:
    \begin{subequations}\allowdisplaybreaks
\begin{align}
    &M_{1B1D}:=\nu^iX^j_A\nu^kX^l_B M_{ijkl} = \left(-\frac{r''}{r} + \frac{k_{\scriptscriptstyle \mathbb{S}}}{n-1}k_{\scriptscriptstyle I} \right) r^2\td \Omega^2_{BD}\\
    &M_{ABCD} =\frac{2m_{\mathrm{H}}(r)}{r^n} r^4(\td \Omega^2_{AB}\td \Omega^2_{BD} -\td \Omega^2_{AD}\td \Omega^2_{CD}) \\
    &M_{11} = (n-1) \left( -\frac{r''}{r} + \frac{k_{\scriptscriptstyle \mathbb{S}}}{n-1}k_{\scriptscriptstyle I} \right)    \\
    & M_{AB} = \left[-\frac{r''}{r} + \frac{   k_{\scriptscriptstyle \mathbb{S}}}{n-1}k_{\scriptscriptstyle I} + (n-2)\frac{2m_{\mathrm{H}}(r)}{r^n}\right] r^2\td \Omega^2_{AB} \\
    &M  = 2 (n-1)\left(-\frac{ r''}{r} +   \frac{ k_{\scriptscriptstyle \mathbb{S}}}{n-1}k_{\scriptscriptstyle I}\right) + (n-1)(n-2)\frac{2m_{\mathrm{H}}(r)}{r^n}\\
    &N_1 =H\left[k_{\scriptscriptstyle I}-\frac{k_{\scriptscriptstyle \mathbb{S}}}{n-1}-r\frac{k'_{\scriptscriptstyle \mathbb{S}}}{n-1}\right]\\
 &N_{1AB}  = -\frac{N_1}{n-1} r^2\td \Omega^2_{AB} \\
 &N_{A1B} = \frac{N_1}{n-1} r^2\td \Omega^2_{AB},
\end{align}
    \end{subequations}
where we recall $m_{\mathrm{H}}$ is given by equation \eqref{eq: m_ch -q^2=F}. After a long computation and using the previous expressions, one can check that \eqref{eq: constraint GB} can be rewritten as
\begin{subequations}\label{eq: mu and J alpha}
    \begin{align}
        &\widetilde{\mu}^{(\alpha)}  =(n-1)(n-2)(n-3)\frac{8m_{\mathrm{H}}(r)}{r^n}\left[\frac{m_{\mathrm{H}}'(r)}{r^{n-1}}-n\frac{ m_{\mathrm{H}}(r)}{2r^n}+\frac{k_{\scriptscriptstyle \mathbb{S}}}{n-1}\frac{N_1}{H}\right]\\
        & (\nu\cdot{}\widetilde{J}^{(\alpha)}) = 2 N_1  \frac{(n-2)(n-3)m_{\mathrm{H}}(r)}{r^{n}}.
   \end{align}
\end{subequations} 
Finally, we compute the derivative of the GB-mass:
\begin{align*}
        m_{\mathrm{GB}}&'(r) = m_{\mathrm{H}}'(r)+m_{\ell}'(r)+4\widetilde{\alpha}\frac{m_{\mathrm{H}}(r)m_{\mathrm{H}}'(r)}{r^n}-2n\widetilde{\alpha}\frac{m_{\mathrm{H}}^2(r)}{r^{n+1}}\overset{\eqref{eq: m'H+m'l}}{\underset{\eqref{eq: M and N}}{=}}\\
        &=\frac{r^{n-1} }{2(n-1)}\left\{M+\frac{n(n-1)}{\ell^2}-2\frac{k_{\scriptscriptstyle \mathbb{S}}}{H}(\nu\cdot{}N)\right\}+\frac{2\widetilde{\alpha}m_{\mathrm{H}}(r)}{r^n}\left(2m_{\mathrm{H}}'(r)-n\frac{m_{\mathrm{H}}(r)}{r}\right)\overset{\eqref{eq: constraints GB}}{=}\\
        &=\frac{r^{n-1} }{2(n-1)}\left\{16\pi\widetilde{\mu}-\alpha\widetilde{\mu}^{(\alpha)}-2\frac{k_{\scriptscriptstyle \mathbb{S}}}{H}(8\pi\nu\cdot{}\widetilde{J}-2\alpha \nu\cdot{}\widetilde{J}^{(\alpha)})\right\}+\frac{2\widetilde{\alpha}m_{\mathrm{H}}(r)}{r^n}\left(2m_{\mathrm{H}}'(r)-n\frac{m_{\mathrm{H}}(r)}{r}\right)\overset{\eqref{eq: mu and J alpha}}{=}\\
        &=\frac{8\pi r^{n-1} }{n-1}\left\{\widetilde{\mu}-\frac{k_{\scriptscriptstyle \mathbb{S}}}{H}(\nu\cdot{}\widetilde{J})\right\}\geq 0.
    \end{align*} 
For the final inequality, we have used the same reasoning as for equation \eqref{eq: m'>0}.
\end{proof}

\subsubsection{Assymptotic limit} \mbox{}\vspace{.5ex}

\noindent Since $\frac{m_{\mathrm{H}}(r)^2 }{r^n}\to0$ as $r\to\infty$, it is clear the same arguments as in the previous section show, assuming the ambient spaces have the appropriate asymptotic behaviour, $m_{\mathrm{GB}}(r)$ approaches the ADM/hyperbolic energies in the limit $r \to \infty$.

\subsection{Penrose inequality} \mbox{}\vspace{.5ex}

\noindent The proof of the Penrose inequality \eqref{eq: PI GB} is almost identical to the charged Penrose inequality. The rigidity statement of Theorem \ref{PI:GB} is similar but requires a bit more work. First, notice that the Einstein-Gauss-Bonnet equations admit static, spherically symmetric solutions \cite{Cai}
\begin{equation}\label{GBSch}
    \td s^2 = -V_{\scriptscriptstyle\mathrm{GB}}(r) \td t^2 + \frac{\td r^2}{V_{\scriptscriptstyle\mathrm{GB}}(r)} + r^2g_{\scriptscriptstyle \mathbb{S}}, \qquad V_{\scriptscriptstyle\mathrm{GB}}(r) = 1 + \frac{r^2}{2\widetilde\alpha} \left(1 \mp \left[ 1 + \frac{8 M \widetilde\alpha}{r^n} - \frac{4\widetilde \alpha}{\ell^2}\right]^{1/2}\right).
\end{equation} We must choose the upper $(-)$ sign so the solutions reduce to Schwarzschild-AdS in the limit $\alpha \to 0$. A $t$-constant spatial hypersurface furnishes a time-symmetric initial data set satisfying \eqref{GB:Hcons} and \eqref{GB:Mcons}. In the case of vanishing cosmological constant ($\ell \to \infty$), the above spacetime metric is asymptotically flat, and one finds the parameter $M$ is equal to the ADM mass defined by \eqref{eq: EADM}. For finite $\ell$, the spacetime is asymptotically AdS$_{n+1}$, but with a rescaled length scale $\ell_{\text{eff}}^2 = \ell^2  (1 + \sqrt{1 - 4\widetilde\alpha / \ell^2})/2$. In our conventions, if we fix $\ell = (1 - \widetilde\alpha)^{-1/2}$ with $\widetilde \alpha \in (0,1/2)$ (this ensures $V_{\scriptscriptstyle\mathrm{GB}}(r) \sim 1 + r^2$ as $r \to \infty$), then the hyperbolic energy \eqref{eq: Ehyp} of \eqref{GBSch} $E_{\text{hyp}} = M/(1 -2\widetilde\alpha)$.

The strategy now, as for the charged Einstein-Maxwell case, is to show that if equality is achieved, the resulting data set can be realized as a spatial hypersurface defined by the level set $t = f(r)$ in the spherically symmetric Gauss-Bonnet black hole spacetime \eqref{GBSch}. Following the method discussed in the Reissner-Nordström case, one finds this is possible provided $f(r)$ is chosen to satisfy \eqref{eq: sol fprime}, with the simple replacement $V_{\scriptscriptstyle\mathrm{RN}} \to V_{\scriptscriptstyle\mathrm{GB}}$. The requirement $m'_{\mathrm{GB}} =0$ implies both $\mu$ and $J$ vanish identically. This demonstrates that the inequality is saturated if and only if the initial data set can be isometrically embedded as a spatial hypersurface in the Gauss-Bonnet vacuum spacetime defined by \eqref{GBSch}.

\end{document}